\documentclass[letterpaper, 10 pt, conference]{ieeeconf}  % Comment this line out if you need a4paper

\IEEEoverridecommandlockouts                              % This command is only needed if
	\usepackage{cite}

	\usepackage[cmex10]{amsmath}
	\usepackage{mathtools}

	\usepackage{array}
	\usepackage{mdwmath}
	\usepackage{mdwtab}

	\usepackage{url}
	
	\usepackage{graphics} % for pdf, bitmapped graphics files
	\usepackage{epsfig} % for postscript graphics files
	\usepackage{amssymb}  % assumes amsmath package installed
	\usepackage{hyperref} % For hyperlinks
	\usepackage{subfig}
	\usepackage{verbatim}
	\usepackage{color}
	\usepackage{multirow} % For multi rows in tables
	\usepackage[makeroom]{cancel}
	\usepackage{caption}

	\newtheorem{claim}{Claim}

	\newtheorem{lemma}{Lemma}

\newcommand{\eq}[1]{\begin{align}#1\end{align}}
\newcommand{\seq}[1]{\begin{subequations}#1\end{subequations}}
\newcommand{\lb}[1]{\left\{ \begin{array}{ll} #1 \end{array} \right.}
\newcommand*\dif{\mathop{}\!\mathrm{d}}
\newcommand{\bm}[1]{\begin{bmatrix}#1\end{bmatrix}}
\newcommand{\E}{\mathbb{E}}

\newcommand{\defeq}{\buildrel\triangle\over =}

\newcommand{\pushright}[1]{\ifmeasuring@ #1 \else\omit\hfill$\displaystyle#1$\fi\ignorespaces}
\newcommand{\pushleft}[1]{\ifmeasuring@ #1 \else\omit$\displaystyle#1$\hfill\fi\ignorespaces}

\begin{document}
	\title{Incentive design for learning in user-recommendation systems with time-varying states}
	% note positions of commas and nonbreaking spaces ( ~ ) LaTex will not break
	% a separate \thanks must be used for each paragraph as LaTex2e's \thanks
	\author{Deepanshu~Vasal, Vijay~Subramanian and Achilleas~Anastasopoulos% <-this % stops a space
	\thanks{The authors are with the Department
	of Electrical Engineering and Computer Science, University of Michigan, Ann
	Arbor, MI, 48105 USA e-mail: { \{dvasal, vgsubram, anastas\} at umich.edu}}% <-this % stops a space
	}

	\maketitle
\begin{abstract}
%We consider a sequential buyers game where exogenously selected strategic users sequentially make a decision to buy or not buy a product that is either good or bad, based on their private observation and publicly available information about decision of the past users. In such games, there are occurrences of informational cascades under certain conditions where a user would discard its private information and base its decision on previous users' actions. In such situations, actions of the users are uninformative for the future users and the learning stops for the team as a whole. For a social objective, it is desirable to avoid bad cascades. We consider an ergodic version of this problem. We formulate the team problem with cooperative users as an instance of decentralized stochastic control problem and characterize the optimal policies. When users are strategic, we design incentives based on team-optimal policies, to be paid to the users by the designer such that users reveal their private signals and learning happens for future users. With the help of these incentives, we show that the gap between the strategic and team objective is small and the overall expected incentives needed is also relatively small. 
We consider the problem of how strategic users with asymmetric information can learn an underlying time-varying state in a user-recommendation system. Users who observe private signals about the state, sequentially make a decision about buying a product whose value varies with time in an ergodic manner. We formulate the team problem as an instance of decentralized stochastic control problem and characterize its optimal policies. With strategic users, we design incentives such that users reveal their true private signals, so that the gap between the strategic and team objective is small and the overall expected incentive payments are also small.
\end{abstract}
\section{Introduction}
In a classical Bayesian learning problem, there is a \textit{single decision maker} who
makes noisy observations of the state of nature and based on these observations
eventually learns the true state. It is well known that through the likelihood
ratio test, the probability of error converges exponentially to zero as the number
of observations increases and the true state is learnt asymptotically.
With the advent of the internet, in today's world, there are many scenarios where strategic agents with different observations (i.e. information sets) interact with each other to learn the state of the system that in turn affects the spread of information in the system. One such scenario was studied by the authors in their seminal paper \cite{BiHiWe92} where they studied the occurrence of fads in a social network, which was later generalized by authors in \cite{SmSo02}. The authors in \cite{BiHiWe92} and \cite{SmSo02} study the problem of learning over a social network where observations are made sequentially by \textit{different decision makers} (users) who act \textit{ strategically} based on their own private information and actions of previous users. It is shown that herding (information cascade) can occur in such a case where a user discards its
own private information and follows the majority action of its predecessors (fads in social networks). As a result, all future users repeat this behavior and a cascade occurs. While a good cascade is desirable, there's a positive probability of a bad cascade that hurts all the users in the community. Thus from a social (i.e. team) perspective, it is highly desirable to avoid such
situations. Avoiding such bad cascades is an active area of research, for example\cite{AcDaLoOz11} and \cite{LeSuBe14} propose alternative learning models that aim at avoiding such bad cascades. In this paper, our goal is to analyze this model and design incentives to avoid bad cascades.

Most of the literature for this problem assumes time-invariant state of the nature. However, there are situations where the state of the nature, for e.g.  popularity of a product, could change over time, as a consequence of endogenous or exogenous factors (for e.g., owing to the entering of a new competitor product or improvement/drop in quality of the product). In this paper we consider a simple scenario where users want to buy a product online. The product is either good or bad (popular or unpopular) and the value of the product (state of the system) is represented by $X_t$, which is changing exogenously via a Markov chain. The state
is not directly observed by the users but each user receives a private noisy
observation of the current state.
Each user makes a decision to either buy or not buy the product, based on its private observation and action profile of all the users before its.

The strategic user wants to maximize its expected value of the product. But its
optimal action could be misaligned with the team objective of maximizing the
expected average reward of the users. Thus the question we seek to address is whether it is possible to incentivize the users to align them with the team objective. To incentivize users to contribute in the learning, we assume that users can
also send reports (at some cost) about their private observations after deciding to buy or to not
buy the product. The idea is similar to leaving a review of the product. Thus
users could be paid to report their observations to enrich the
information of the future participants. Our objective is to use principles of
mechanism design to construct the appropriate payment transfers (taxes/subsidies). Although, our approach deviates from general principles of mechanism design for solution of the game problem to \emph{exactly} coincide with the team problem. However, this analysis could provide the bounds on the gap and an acceptable practical design. 

%\subsection{Notation}
We use uppercase letters for random variables and lowercase for their realizations. We use notation $a_{t:t'}$ to represent vector $(a_t, a_{t+1}, \ldots a_{t'})$ when $t'\geq t$ or an empty vector if $t'< t$. We denote the indicator function of any set $A$ by $I_{A}(\cdot)$.
For any finite set $\mathcal{S}$, $\mathcal{P}(\mathcal{S})$ represents space of probability measures on $\mathcal{S}$ and $|\mathcal{S}|$ represents its cardinality. We represent the set of real numbers by $\mathbb{R}$. We denote by $P^g$ (or $E^g$) the probability measure generated by (or expectation with respect to) strategy profile $g$. All equalities and inequalities involving random variables are to be interpreted in \emph{a.s.} sense. We use the terms users and buyers interchangeably.

The paper is structured as follows. In section~\ref{sec:Model}, we present the model. In section~\ref{sec:TeamProblem}, we formulate the team problem as an instance of decentralized stochastic control and characterize its optimal policies. In section~\ref{sec:IncentiveDesign}, we consider the case with strategic users and design incentives for the users to align their objective with team objective. We conclude in section~\ref{sec:Conclusion}.

\section{Model}
\label{sec:Model}
We consider a discrete-time dynamical system over infinite horizon. There is a product whose value varies over time as (a slowly  varying) discrete time Markov process $(X_t)_t$, where $X_t$ takes value in the set $\{0,1\}$; 0 represents that product was bad (has low intrinsic value) and 1 represents and product is good (has high intrinsic value).  
\seq{
\label{eq:xMarkov}
\eq{
P(x_1) &= \hat{Q}(x_1)\\
P(x_t|x_{1:t-1}) &= Q_x(x_t|x_{t-1}),
}
}
such that $Q_x(x_{t}|x_{t-1})=\epsilon$ if $x_t\neq x_{t-1}$, for $0<\epsilon<1$.

There are countably infinite number of exogenously selected, selfish buyers that act sequentially and exactly once in the process. Buyer $t$ makes a noisy observation of the value of the product at time $t$, $v_t \in \mathcal{V} \defeq \{ 0,1\}$, through a binary symmetric channel with crossover probability $p$ such that these observations are conditionally independent across users given the system state (i.e. noise is i.i.d.) i.e. $P(v_t|x_{1:t}v_{1:t-1}) = Q_v(v_t | x_t ) =p$ if $v_t\neq x_{t}$. Based on actions of previous buyers and its private observation buyer $t$ takes two actions: $a_t \in \mathcal{A} \defeq \{0,1\}$, which correspond to either buying or not buying the good, and $b_t \in \mathcal{B} \defeq \{*,1\}$ where * represents not reporting its observation and 1 represent reporting truthfully. Based on these actions and the state of the system, the buyer gets reward $R(x_t,a_t,b_t)$ where
\eq{
 &R(x_t,a_t, b_t) \nonumber \\
  &= -c \cdot I(b_t = 1) +
   \lb{\;\;\ 1/2, \;\;\; \hfill x_t = 1,a_t = 1\\
        -1/2, \;\;\; \hfill  x_t = 0, a_t = 1  \\
      \;\;\;\;0, \;\;\; \hfill a_t = 0 },
}
where $c$ is cost of reporting its observation truthfully. The actions are publicly observed by future buyers whereas the observations $(v_t)_t$ are private information of the buyers.

\section{Team problem}
\label{sec:TeamProblem}

In this section we study the team problem where the buyers are cooperative and want to maximize the expected average reward per unit time for the team. At time $t$, buyer $t$'s information consists of its private information $v_t$ and publicly available information $a_{1:t-1},b_{1:t-1}$. It takes action $a_t,b_t$ though a (deterministic) policy $g_t: \mathcal{A}^{t-1}\times \mathcal{B}^{t-1}  \times \mathcal{V} \to \mathcal{A}\times \mathcal{B} $ as 
\eq{
(a_t,b_t) = g_t(a_{1:{t-1}},b_{1:{t-1}},v_t).
}
%\subsection{Structural result 1: Can be reduced to a POMDP}
The objective as a team (or for a social planner) is to maximize the expected average reward per unit time for all the users i.e.
\eq{
J\defeq  \sup_g \limsup_{\tau\to \infty} \frac{1}{\tau} \sum_{t=1}^{\tau} \E^g\{R(X_t,A_t, B_t)\}.
}

Since the decision makers (i.e. the buyers) have different information sets, this is an instance of a decentralized stochastic control problem. We use techniques developed in \cite{NaMaTe11} to find structural properties of the optimal policies. Specifically, we equivalently view the system through the perspective of a common agent that observes at time $t$, the common information $a_{1:t-1},b_{1:t-1}$ and takes action $\gamma_t : \mathcal{V} \to \mathcal{A}\times \mathcal{B}$, which is a partial function that, when acted upon buyer's private information $v_t$, generates its action $(a_t,b_t)$. The common agent's actions $(\gamma_t)_t$ are taken through common agent's strategy $\psi = (\psi)_t$ as $\gamma_t = \psi_t[a_{1:t-1},b_{1:{t-1}}]$ where $\psi_t: \mathcal{A}^{t-1}\times \times \mathcal{B}^{t-1} \to \left( \mathcal{V} \to \mathcal{A}\times \mathcal{B}\right)$. The corresponding common agent's problem is 
\eq{
J^c\defeq  \sup_{\psi} \limsup_{\tau\to \infty} \frac{1}{\tau} \sum_{t=1}^{\tau} \E^{\psi}\{ R(X_t,A_t, B_t)\}.
}
This procedure transforms the original decentralized stochastic control problem of buyers to a centralized stochastic control problem of the common agent. Thus an optimal policy of common agent can be translated to optimal policy for the buyers. In order to characterize common agent's optimal policies, we find an information state for the common agent's problem. We define a belief state $\pi_t$ at time $t$ as a probability measure on current state of the system given the common information i.e. $\pi_t(x_t) \defeq P^{\psi}(x_t|a_{1:t-1} b_{1:t-1} \gamma_{1:t})$. The following lemma shows that the common agent faces a Markov decision problem (MDP). 

\begin{lemma}
\label{lemma:CMP1_ch4}
$(\Pi_t, \Gamma_t)_t$ is a controlled Markov process with state $\Pi_t$ and action $\Gamma_t$ such that 
\seq{
\eq{
&P^{\psi}(\pi_{t+1} | \pi_{1:t} \gamma_{1:t}) = P(\pi_{t+1} | \pi_t \gamma_t)\label{eq:Pi_update_ch4}\\
&\E^{\psi}\{ R(X_t,A_t, B_t)|a_{1:t-1}b_{1:t-1} \gamma_{1:t} \}&  \nonumber \\
=& \E\{ R(X_t,A_t, B_t)| \pi_t\gamma_{t}\}\\
=:& \hat{R}(\pi_t, \gamma_t)
} 
}
and there exists an update function $F$, independent of $\psi$ such that $\pi_{t+1} = F (\pi_t, \gamma_t, a_t,b_t )$. 
\end{lemma}
\begin{proof}
See Appendinx.
\end{proof}

Lemma~\ref{lemma:CMP1_ch4} implies that for common agent's problem, it can summarize the common information $a_{1:t-1},b_{1:t-1}$ in the belief state $\pi_t$. Furthermore there exists an optimal policy for the common agent of the form $\theta_t : \mathcal{P}(\mathcal{X}) \to \left( \mathcal{V} \to \mathcal{A}\times \mathcal{B}\right)$ that can be found as solution of the following dynamic programming equation in the space of public beliefs $\pi_t$ as, $\forall \pi, \gamma^* = \theta[\pi]$ is the maximizer in the following equation  
\eq{ \rho + V(\pi) = \max_{\gamma} \;\; \hat{R}(\pi, \gamma) + \E\{V(\Pi')|\pi \gamma \} \label{eq:DP1_ch4},
}
where the distribution of $\pi'$ is given through the kernel $P(\cdot|\pi \gamma)$ in~(\ref{eq:Pi_update_ch4}) and $\rho \in \mathbb{R}, V: \mathcal{P}(\mathcal{X})\to \mathbb{R}$ are solution of the above fixed point equation. Based on this public belief $\pi_t$ and its private information $x_t$, each user $t$ takes actions as
\eq{
(a_t,b_t) = m_t(\pi_t,v_t) = \theta_t[\pi_t] (v_t) .
} 

We note that since states, actions and observations belong to a binary set, there are sixteen partial functions $\gamma$ possible that are shown in Table~\ref{tab:Gamma} below where $\gamma = \bm{\gamma(v_t=0) \\ \gamma(v_t=1)}=\bm{a_t,b_t (v_t=0) \\ a_t,b_t ( v_t=1)} $. Since the common belief is updated as $\pi_{t+1} = F(\pi_t, \gamma, \gamma(v_t))$ and $v_t$ is binary valued, there exist two types of $\gamma$ functions: learning ($\gamma^{L}$) and non-learning ($\gamma^{NL}$). $\gamma^L$ leads to update of belief through $F(\cdot)$ in (\ref{eq:Pi_update_ch4}) that is informative of the private observation $v_t$, whereas $\gamma^{NL}$ leads to uninformative update of belief. Eight of them are dominated in reward for example $v_t$ need not be reported if it is revealed through $a_t$, or if it can be revealed indirectly by absence of reporting. 

\begin{table}[htbp]
{
%\begin{small}
\begin{center}
\caption{}
\label{tab:Gamma}
\begin{tabular}{|lllllll|}
\hline
$\gamma^L$ & $\bm{0,*\\ 1,*}$ & $\bm{1,*\\ 0,*}$  & $\bm{1,1\\ 1,*}$ & $\bm{1,*\\ 1,1}$ & $ \bm{0,1\\ 0,*}$ & $\bm{0,*\\ 0,1}$   \\
 & \xcancel{$\bm{0,1\\ 1,1}$}  & \xcancel{$\bm{1,1\\ 0,1}$} & \xcancel{$\bm{0,1\\ 1,*}$} &  \xcancel{$\bm{1,1\\ 0,*}$} & \xcancel{$\bm{0,*\\ 1,1}$} &  \xcancel{$\bm{1,*\\ 0,1}$}  \\
  & \xcancel{$\bm{0,1\\ 0,1}$} & \xcancel{$\bm{1,1\\ 1,1}$} & & & & \\
\hline
$\gamma^{NL}$ & $\bm{0,*\\ 0,*}$ & $\bm{1,*\\ 1,*}$ & & & & \\
\hline
\end{tabular}
\end{center}
%\end{small}
}
\end{table}

\section{Game problem}
\label{sec:IncentiveDesign}

We now consider the case when the buyers are strategic. As before, buyer $t$ observes public history $a_{1:t-1},b_{1:t-1}$ and its private observation $v_t$ and thus takes its actions as $(a_t,b_t) = g_t(a_{1:t-1},b_{1:t-1},v_t)$. Its objective is to maximize its expected reward
\eq{
    J_t &= \max_{g_t} \; \E^{g}\{ R(X_t,A_t, B_t)\}.
}
Since all buyers have different information, this defines a dynamic game with asymmetric information. An appropriate solution concept is Perfect Bayesian Equilibrium (PBE)~\cite{OsRu94} that requires specification of an assessment $(g_t^*,\mu_t^*)_t$ of strategy and belief profile where $g_t^*$ is the strategy of buyer $t$, $g^*_t : \mathcal{A}^{t-1} \times \mathcal{B}^{t-1}\times\mathcal{V} \to \mathcal{P}(\mathcal{A} \times\mathcal{B})$, and $\mu_t^*$ is a belief as a function of buyer $t$'s history on the random variables not observed by it till time $t$ i.e. $\mu^*_t : \mathcal{A}^{t-1} \times \mathcal{B}^{t-1}\times\mathcal{V} \to \mathcal{P}(\mathcal{X}^t \times\mathcal{V}^t)$. In general, finding a PBE is hard~\cite{OsRu94} since it involves solving a fixed point equation in strategies and beliefs that are function of histories although there are few cases where there exists an algorithm to find them~\cite{NaGuLaBa14,VaSuAn15arxiv}. For this problem, since users act exactly once in the game and are thus myopic, it can be found easily in a forward inductive way, as in~\cite{BiHiWe92,SmSo02}. Moreover, a belief on $X_t$, $\mu^*_t(x)\defeq P^{g^*}(X_t=x|a^{t-1},b^{t-1},v_t), x\in\{0,1\}$ is sufficient and any joint belief consistent with $\mu^*_t(x)$ along with equilibrium strategy profile $g^*$ constitute a PBE.
For any history, users compute a belief equilibrium strategy depending on $v_t$ and $\pi_t$ as
        \eq{ %\text{ or equivalently }
\gamma^*_t = \phi[\pi_t] = \arg\max_{\gamma_t} \hat{R}(\pi_t,\gamma_t) \label{eq:StOpt}
    }
With $\phi[\cdot]$ defined through (\ref{eq:StOpt}), for every history $(a_{1:t-1},b_{1:t-1},v_t)$, $\pi_t$ is updated using forward recursion through $\pi_{t+1} = F(\pi_t, \phi(\pi_t), a_t,b_t)$ and equilibrium strategies are generated as $g_t^*(a_{1:t-1},b_{1:t-1},v_t)=\phi[\pi_t](v_t)$. Finally the beliefs $\mu_t^*$ can be easily derived from $\pi_t$ and private information $v_t$ through Bayes rule.
 
 We numerically solve (\ref{eq:DP1_ch4}) using value iteration to find team optimal policy, shown in Figure~\ref{fig:Dec_ch4}, for parameters $p=0.2,\epsilon =0.001$ and $c=0.05$. For the same parameters, Figure~\ref{fig:St_ch4} shows optimal policy for a strategic user that solves (\ref{eq:StOpt}). 

\begin{figure}[htbp]
\centering
\captionsetup{justification=centering}
\includegraphics[width=3.7in]{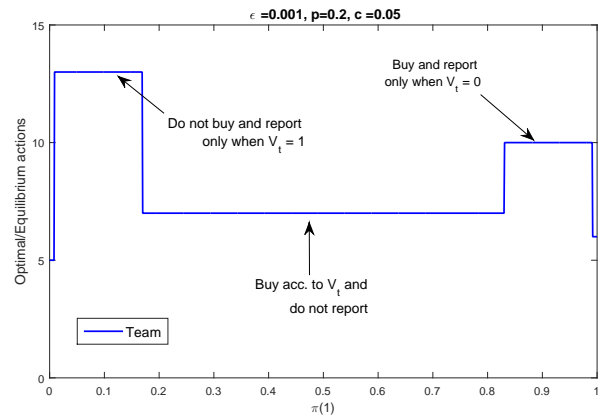}
\caption{Decentralized team optimal policy}
\label{fig:Dec_ch4}
\end{figure}
\begin{figure}[htbp]
\centering
\captionsetup{justification=centering}
\includegraphics[width=3.7in]{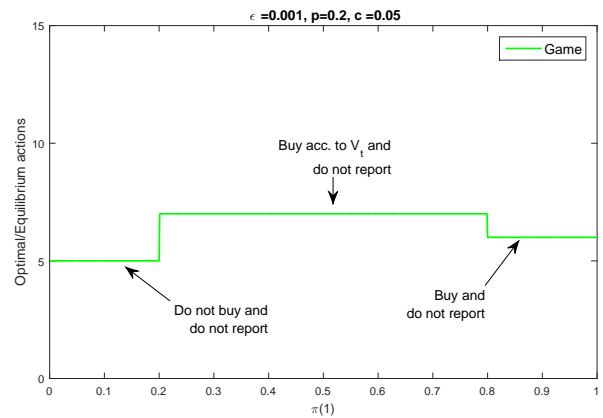}
\caption{Strategic optimal policy}
\label{fig:St_ch4}
\end{figure}

\subsection{Incentive design for strategic users}
%For the game problem, each strategic user observes $a_{1:t-1},b_{1:t-1}$ and $v_t$ and maximizes its expected instantaneous reward $\E\{R(X_t,A_t,B_t)|a_{1:t-1},b_{1:t-1} \gamma_t \}$ that is equivalent to $\hat{R}(\pi_t,\gamma_t)$ for a given equilibrium policy profile.

Our goal is to align each buyers' objective with the team objective. In order to do so, we introduce incentives (tax or subsidy) for user $t$, $t:\mathcal{P}(\mathcal{X}) \times \mathcal{A} \times \mathcal{B} \to \mathbb{R} $ such that its effective reward is given by $\hat{R}(\pi_t,\gamma_t) - t(\pi_t, a_t, b_t)$. 

We first note that a user can not internalize social reward through incentives as is done in a pivot mechanism~\cite{Vi61,Cl71,Gr73,BeVa10}, i.e. there does not exist an incentive mechanism such that the following equation could be true
\eq{
 \hat{R}(\pi,\gamma) - t(\pi, a, b)&= \hat{R}(\pi, \gamma) + \E\{V(\Pi')|\pi \gamma \} \label{eq:EqObj} \\
\text{i.e. \hspace{30pt}}  t(\pi, a, b) &= - \E\{V(\Pi')|\pi \gamma \}
}
for $V(\cdot)$ defined in (\ref{eq:DP1_ch4}) and the distribution of $\pi'$ is given through the kernel $P(\cdot|\pi \gamma)$ in~(\ref{eq:Pi_update_ch4}). The left side of (\ref{eq:EqObj}) is buyers' effective reward and right side is the objective of the team problem as in (\ref{eq:DP1_ch4}). Such a design is not feasible because while $t(\cdot)$ can depend only on public observations ($\pi,a,b$), the second term in the RHS of (\ref{eq:EqObj}) depends on $\gamma$ as well which is not observed by the designer.
% \textcolor{red}{This is a strong statement and I am not very sure if this is true.}

We observe in Figures~\ref{fig:Dec_ch4},~\ref{fig:St_ch4} that team optimal policy coincides with the strategic optimal policy for a significant range of $\pi(1)$. Let $\mathcal{S}$ be the set consisting of $\pi(1)$ where the team optimal policy coincides with the strategic optimal policy and $\mathcal{S}^c$ be the complement set. In order to align the two policies, we consider the following incentive design such that a user is paid $c$ units by the system planner whenever the public belief $\pi(1)$ belongs to the set $\mathcal{S}^c$ and user reports its observation, 
\eq{
t(\pi, a_t, b_t) = -c\cdot I(\pi(1) \in \mathcal{S})I(b_t = 1).
}
These payments are made after any report for enforcement purposes. This is agreed upon, i.e., system planner commits to this.
With these incentives, the optimal policy of the strategic user is shown in Figure~\ref{fig:Mch_ch4}. Figure~\ref{fig:Cost_ch4} compares the time average reward achieved through these policies, found through numerical results. This shows that the gap between the team objective and the one with incentives is small. Intuitively, this occurs because the buyers learn the true state of the system relatively quickly (exponentially fast) compared to the expected time spent by the Markov process $X_t$ in any state. Equivalently, the time spent by the process $(\Pi_t(1))_t$ in the set $\mathcal{S}^c$ is small. Yet it is crucial for the social objective that learning occurs in this region.
Also in Figure~\ref{fig:Cost_ch4}, the gap between the mechanism (including incentives) and the mechanism where incentives are subtracted signifies the expected average payment made by the designer, which is relatively small.
\begin{figure}[htbp]
\centering
\captionsetup{justification=centering}
\includegraphics[width=3.7in]{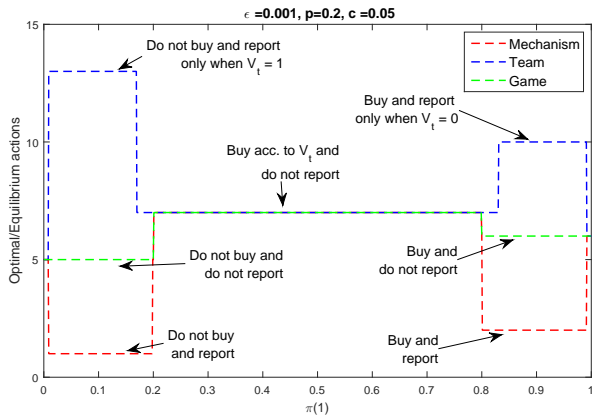}
\caption{Strategic optimal policy with incentives}
\label{fig:Mch_ch4}
\end{figure}
\begin{figure}[htbp]
\centering
\captionsetup{justification=centering}
\includegraphics[width=3.7in]{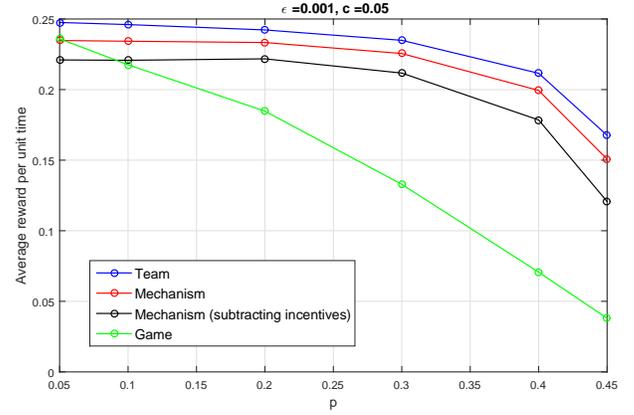}
\caption{Expected time average cost comparison for different policies }
\label{fig:Cost_ch4}
\end{figure}

\section{Conclusion}
\label{sec:Conclusion}

We considered a sequential buyers game where a countable number of strategic buyers buy a product exactly once in the game. The value of the product is modeled as a Markov process and buyers privately make noisy observation of the value. We model the team problem as an instance of decentralized stochastic control problem and characterize structure of optimum policies. When users are strategic, it is modeled as a dynamic game with asymmetric information. We show that for some set $\pi_t \in \mathcal{S}$ that occurs with high probability, the strategic optimal policy coincides with the team optimal policy. Thus only outside this set, i.e., when
$\pi_t \in \mathcal{S}^c$, buyers need to be incentivized to report their observations so that higher average rewards can be achieved for the whole team. Since numerically $\mathcal{S}^c$ occurs with low probability, the expected incentive payments are low. However, even though infrequent, these incentives help in the learning for the team as a whole, specifically for the future users. This suggests that using such a mechanism for the more general case could be a useful way to bridge the gap between strategic and team objectives.

Future work involves characterizing team-optimum policies analytically and studying the resulting social utility through approximations or bounds on the induced Markov chain statistics. This would also characterize the gain from introducing ``structured'' incentives. Finally, incentives designs could be studied that minimize total expected incentives and guarantee voluntary participation.

%Guided by these initial results, we plan to design and analyze incentives for a more general class of problems. These designs aim at incentivizing users to align their strategies with the team optimum. The analysis proceeds by obtaining bounds on the steady-state probability distribution of the underlying Markov processes. These bounds can subsequently be utilized to bound the gap between team and strategic objectives as well as the expected incentives. There are other variants of incentives that we propose to study, for e.g. instead of a fixed incentive, a lottery is held among the users who agree to share their observations. This provides more incentive to users to report when there is less public information about the true state and vice versa, and thus could further minimize the expected incentives.

\appendix
\section{}
\label{app:A}
\begin{claim}
There exists an update function $F$, independent of $\psi$ such that $\pi_{t+1} = F (\pi_t, \gamma_t, a_t,b_t )$.
\end{claim}
\begin{proof}
	 Fix $\psi$
	 \begin{subequations}
	\eq{
	 \pi_{t+1}(x_{t+1}) =& P^{\psi}(x_{t+1} |a_{1:t}b_{1:t}  \gamma_{1:t})\\
	=& \sum_{x_t} P^{\psi}( x_{t+1},x_t |a_{1:t}b_{1:t}  \gamma_{1:t})\\
	=& \sum_{x_t}P^{\psi}(x_t |a_{1:t}b_{1:t}  \gamma_{1:t})\hat{Q}(x_{t+1}| x_t) 
	}
\end{subequations}
	Now,
\begin{subequations}
	\eq{
	&P^{\psi}(x_t |a_{1:t}b_{1:t}  \gamma_{1:t}) \nonumber \\
	&= \frac{P^{\psi}(x_t, a_t,b_t |a_{1:t-1}b_{1:t-1}, \gamma_{1:t})}{\sum_{\hat{x}_t} P^{\psi}(\hat{x}_t, a_t,b_t  |a_{1:t-1}b_{1:t-1}, \gamma_{1:t})}\\
	&=P^{\psi}(x_t |a_{1:t-1}b_{1:t-1}, \gamma_{1:t}) \times \nonumber \\
	& \;\;\;\;\;\;\;\;\frac{ \sum_{v_t} P^{\psi}( a_t,b_t v_t|a_{1:t-1}b_{1:t-1}, \gamma_{1:t}, x_t)}{\sum_{\hat{x}_t} P(\hat{x}_t, a_t,b_t  |a_{1:t-1}b_{1:t-1}, \gamma_{1:t})}\\
	&= \frac{{P^{\psi}(x_t |a_{1:t-1}b_{1:t-1}, \gamma_{1:t-1})}{ \sum_{v_t} I_{\{\gamma_{t}(v_t)\}}(a_t,b_t) Q_v(v_t|x_t)}}{\splitfrac{\sum_{\hat{x}_t} P^{\psi}(\hat{x}_t |a_{1:t-1}b_{1:t-1}, \gamma_{1:t-1})}{\sum_{v_t} I_{\{\gamma_{t}(v_t)\}}(a_t,b_t)Q_v(v_t|\hat{x}_t)} }\label{eq:appH1}
	}
\end{subequations}
	where first part in numerator in (\ref{eq:appH1}) is true since given policy $\psi$, $\gamma_{t}$ can be computed as  $\gamma_t = \psi_t(a_{1:t-1}b_{1:t-1})$.
	
We conclude that
	\eq{
	&P(x_t |a_{1:t}, \gamma_{1:t}) \nonumber \\
	=& \frac{\pi_t(x_t)\sum_{v_t} I_{\{\gamma_{t}(v_t)\}}(a_t,b_t) Q_v(v_t|x_t)}{\sum_{\hat{x}_t} \pi_t(\hat{x}_t) \sum_{v_t} I_{\{\gamma_{t}(v_t)\}}(a_t,b_t) Q_v(v_t|\hat{x}_t)},
	}
	thus,
	\eq{
	\pi_{t+1}  = F(\pi_t, \gamma_t, a_t,b_t)
	}
	where $F$ is independent of policy $\psi$.
		\end{proof}

\begin{claim}
$(\Pi_t, \Gamma_t)_t$ is a controlled Markov process with state $\Pi_t$ and action $\Gamma_t$ such that 
\eq{
P^{\psi}(\pi_{t+1} | \pi_{1:t} \gamma_{1:t}) &= P(\pi_{t+1} | \pi_t \gamma_t)
}
\eq{
&\E^{\psi}\{ R(X_t,A_t,B_t)| \gamma_{1:t} a_{1:t-1}b_{1:t-1}\} \nonumber \\
&= \E\{ R(X_t,A_t,B_t)|\gamma_{t} \pi_t\}\\
&=: \hat{R}(\pi_t, \gamma_t)
} 
\end{claim}
\begin{proof}
	\begin{subequations}
	\eq{
	&P^{\psi} (\pi_{t+1}|\pi_{1:t}, \gamma_{1:t}) \nonumber \\
	&= \sum_{a_t,b_t} P^{\psi} (\pi_{t+1}, a_t,b_t|\pi_{1:t}, \gamma_{1:t} )\\
	&= \sum_{a_t,b_t} \textbf{1}_{\{ F(\pi_t,\gamma_t, a_t,b_t)\}}(\pi_{t+1}) \sum_{v_t} P^{\psi} (a_t,b_t v_t|\pi_{1:t}, \gamma_{1:t} )\\
	&= \sum_{a_t,b_t, x_t} \textbf{1}_{\{ F(\pi_t,\gamma_t, a_t,b_t)\}}(\pi_{t+1})P^{\psi} (x_t|\pi_{1:t}, \gamma_{1:t})  \nonumber \\
	& \;\;\;\;\;\;\;\;\; \sum_{v_t} I_{\{\gamma_{t}(v_t)\}}(a_t,b_t) Q_v(v_t|x_t) \\
	&= \sum_{a_t,b_t, x_t} \pi_t(x_t) \textbf{1}_{\{ F(\pi_t,\gamma_t, a_t,b_t)\}}(\pi_{t+1}) \nonumber \\
	& \;\;\;\;\;\;\;\;\;   \sum_{v_t} I_{\{\gamma_{t}(v_t)\}}(a_t,b_t) Q_v(v_t|x_t)  \\
	&= P (\pi_{t+1}|\pi_{t}, \gamma_{t})
	}
	\end{subequations}
	\begin{subequations}
	\eq{
	&\mathbb{E} ( R(X_t, A_t,B_t)| \pi_{1:t},\gamma_{1:t}) \nonumber \\
	&= \sum_{x_t, a_t,b_t v_t} R(x_t, a_t,b_t) P(x_t,a_t,b_t,v_t| \pi_{1:t},\gamma_{1:t})\\
	&= \sum_{x_t, a_t,b_t} R(x_t, a_t,b_t) P(x_t| \pi_{1:t},\gamma_{1:t})\nonumber \\
	&\;\;\;\;\;\;\;\;\;\; \sum_{v_t} I_{\{\gamma_{t}(v_t)\}}(a_t,b_t) Q_v(v_t|x_t) \\
	&= \sum_{x_t, a_t,b_t} R(x_t, a_t,b_t) \pi_t(x_t)\nonumber \\
	&\;\;\;\;\;\;\;\;\;\; \sum_{v_t} I_{\{\gamma_{t}(v_t)\}}(a_t,b_t) Q_v(v_t|x_t) \\
	&= \hat{R} (\pi_t, \gamma_t)
	}
	\end{subequations}
	\end{proof}

\bibliographystyle{IEEEtran}
%\bibliography{deepanshu15}

\vfill
	
        \end{document}